\documentclass[reqno]{amsart}

\usepackage{amssymb,latexsym}
\usepackage{eucal}
\usepackage{enumerate}

\theoremstyle{plain}
\newtheorem{theorem}{Theorem}[section]
\newtheorem{proposition}{Proposition}[section]

\theoremstyle{definition}
\newtheorem{definition}{Definition}[section]

\numberwithin{equation}{section}

\RequirePackage[OT1]{fontenc}
\RequirePackage{amsthm,amsmath}
\RequirePackage[numbers]{natbib}
\RequirePackage[colorlinks,citecolor=blue,urlcolor=blue]{hyperref}
\RequirePackage{graphics}
\RequirePackage{subfigure}
\RequirePackage{epstopdf}
\usepackage{blindtext}
\usepackage[utf8]{inputenc}
\usepackage{graphics,graphicx,epsfig}
\usepackage{pdflscape}
\usepackage{epstopdf}
\usepackage{grffile}
\usepackage{subfigure}

\makeatletter
\newcommand{\leqmode}{\tagsleft@true}
\newcommand{\reqmode}{\tasleft@false}
\makeatother

\begin{document}

\title[K-Groups]{K-groups: A Generalization of K-means Clustering}

\author{Songzi Li}
\address{Biostatistics and Programming Department \\
         Pharmaceutical Product Deparment, LLC\\
         Wilmington, NC 28401}
\email{lisongzi24601@gmail.com}

\author{Maria L. Rizzo}
\address{Dept. of Mathematics \& Statistics\\
    Bowling Green State University\\
    Bowling Green, OH 43403
}

\date{July 25, 2017}

\email{mrizzo@bgsu.edu}

\begin{abstract}
  We propose a new class of distribution-based clustering algorithms, called k-groups, based on energy distance between samples. The energy distance clustering criterion assigns observations to clusters according to a multi-sample energy statistic that measures the distance between distributions. The energy distance determines a consistent test for equality of distributions, and it is based on a population distance that characterizes equality of distributions.  The k-groups procedure therefore generalizes the k-means method, which separates clusters that have different means. We propose two k-groups algorithms: k-groups by first variation; and k-groups by second variation. The implementation of k-groups is partly based on Hartigan and Wong's algorithm for k-means. The algorithm is generalized from moving one point on each iteration (first variation) to moving $m$ $(m > 1)$ points. For univariate data, we prove that Hartigan and Wong's k-means algorithm is a special case of k-groups by first variation. The simulation results from univariate and multivariate cases show that our k-groups algorithms perform as well as Hartigan and Wong's k-means algorithm when clusters are well-separated and normally distributed. Moreover, both k-groups algorithms perform better than k-means when data does not have a finite first moment
  or data has strong skewness and heavy tails. For non--spherical clusters, both k-groups algorithms performed better than k-means in high dimension, and k-groups by first variation is consistent as dimension increases. In a case study on dermatology data with 34 features, both k-groups algorithms performed better than k-means.
\end{abstract}


\maketitle

\section{Introduction}
Cluster analysis is one of the core topics of data mining and has many application domains such as astronomy, psychology, market research and bioinformatics. Clustering is a fundamental tool in unsupervised methods of data mining, where it is used to group similar objects together without using external information such as class labels. In general, there are two purposes for using cluster analysis: understanding and utility \cite{milligan1996clustering}. Understanding in cluster analysis means to find groups of objects that share common characteristics. Utility of cluster analysis aims to abstract the representative objects from objects in the same groups. The earliest research on cluster analysis can be traced back to 1894, when Karl Pearson used the moment matching method to determine the mixture parameters of two single-variable components \cite{pearson1894contributions}. There are various clustering algorithms, each algorithm with its own advantages in specific situations. In this paper we propose new clustering methods   \emph{$k$-groups} which generalize and extend the well known and widely applied $k$-means cluster analysis method. Our cluster distance is based on a characterization of equality between distributions, and it applies in arbitrary dimension. It generalizes $k$-means, which
separates clusters by differences in means. The $k$-groups cluster distance \cite{ls2015}, which is based on \emph{energy distance} \cite{szekely2004testing,rizzo2010disco,sr2013b}, separates clusters by differences in distributions.

\subsection{K-means}

$K$-means is a prototype-based algorithm which uses the cluster mean as the centroid, and assigns observations to the cluster with the nearest centroid.
Let $D = \{x_{1}, \ldots, x_{n}\} \subset R^{m} $ the data set to be clustered, and $P = \{ \pi_{1}, \ldots, \pi_{K} \}$, a partition of $D$, where $K$ is the number of clusters set by the user.  Thus we have $\cup_{i} \pi_{i}=D ,$ and $\pi_{i} \cap \pi_{j}= \emptyset$ if $i\neq j$. The symbol $\omega_{x}$ denotes the weight of $x$,  $n_{k}$ is the number of data objects assigned to cluster $\pi_{k}$, and  $c_{k}=\sum_{x\in \pi_{k}} \frac{\omega_{x} x}{n_{k}}$ represents the centroid of cluster $\pi_{k}$, $1 \le k \le K$. The function $d(x,y)$ is a symmetric, zero-diagonal dissimilarity function that measures the distance between data objects $x$ and $y$. The $k$-means clustering objective is
\begin{equation}
    min_{c_{k}, 1\le k\le K}\sum_{k=1}^K\sum_{x\in \pi_{k}} \omega_{x} d(x,c_{k}).
\end{equation}
Implementing a $k$-means algorithm is equivalent to a global minimum problem which is computationally difficult (NP-hard). An early  algorithm proposed by Stuart Lloyd in 1957 \cite{lloyd1982least} was commonly applied. A more efficient version was proposed and published in Fortran by Hartigan and Wong in 1979 \cite{hartigan1979algorithm}. The distance or dissimilarity function $d(x,y)$ is one of the important factors that influences the performance of $k$-means. The most commonly used distance functions are Euclidean quadratic distance, spherical distance, and Kullback-Leibler Divergence \cite{tan2006cluster}. Each choice determines a cluster distance function. In this paper, we propose a new cluster distance function: \emph{Energy Distance}.

\subsection{Energy Distance}

Sz{\'e}kely proposed \emph{energy distance}  \citep{TR2002,TR89},  as  a statistical distance between samples of observations.
For an overview of methods based on energy distance see e.g.\ \cite{sr2013b,sr2017}. The concept is based on the notion of Newton's gravitational potential energy, which is a function of the distance between two bodies in a gravitational space.
\begin{definition}
\emph{Energy Distance}. The energy distance between the $d$-dimensional independent random variables $X$ and $Y$ is defined as
\begin{equation*}
\mathcal{E}(X,Y)=2E|X-Y|_{d}-E|X-X'|_d-E|Y-Y'|_d,
\end{equation*}
where $E|X|_d<\infty$, $E|Y|_d<\infty$, $X'$ is an independent and identically distributed (iid) copy of $X$, and $Y'$ is an iid copy of Y. \label{energy distance}
\end{definition}
Here and throughout, $| \cdot |_d$ denotes Euclidean distance in $R^d$, and we omit $d$ when the dimension is clear in context. A primed random variable denotes an iid copy; that is, $X$ and $X'$ are iid.

Let $F(x)\,\mathrm{and}\,G(x)$ be the cumulative distribution functions, and $\hat{f}(t)$ and $\hat{g}(t)$ be the characteristic functions of independent random variables $X$ and $Y$, respectively. The following definition is introduced in Sz{\'e}kely (2002), and in Sz{\'e}kely and Rizzo (2005).

\begin{definition}\label{energy distance alpha}
Let $X$ and $Y$ be independent $d$-dimensional random variables with characteristic functions $\hat{f},\hat{g}$, respectively, and $E|X|^{\alpha} < \infty$, $E|Y|^{\alpha} < \infty$ for some $0<\alpha<2$. The energy distance between $X$ and $Y$ is defined as
\begin{align}
\label{energy th1} \mathcal{E}^{\alpha}(X,Y)=& 2E|X-Y|_d^{\alpha}-E|X-X'|_d^{\alpha}-E|Y-Y'|_d^{\alpha} \\ \notag
             =& \frac{1}{C(d,\alpha)}\int_{R^d} \frac{|\hat{f}(t)-\hat{g}(t)|^2}{|t|_d^{d+\alpha}} dt,
\end{align}
where $0<\alpha<2$, $|\cdot|$ is the complex norm, and
\begin{equation*}
C(d,\alpha)=2\pi^{\frac{d}{2}}\frac{\Gamma\left(1-\frac{\alpha}{2}\right)}{\alpha 2^{\alpha}\Gamma\left(\frac{d+\alpha}{2}\right)}
\end{equation*}
\end{definition}

The following theorem (\cite[Theorem 2]{TR2002}, \cite[Theorem 2]{szekely2005new}]) establishes that energy distance between
random variables characterizes equality in distribution.
\begin{theorem} \label{energy distance alpha p}
For all $0<\alpha<2$,
$$
\mathcal{E}^{\alpha}(X,Y)= 2E|X-Y|_d^{\alpha}-E|X-X'|_d^{\alpha}-E|Y-Y'|_d^{\alpha} \geq 0,
$$
with equality to zero if and only if $X$ and $Y$ are identically distributed.
\end{theorem}

Note that when $\alpha=2$, we have
\begin{equation*}
2E|X-Y|^{2}-E|X-X'|^{2}-E|Y-Y'|^{2} = 2|E(X) - E(Y)|^2,
\end{equation*}
which measures the squared distance between means.
Hence, the above characterization does not hold for $\alpha =2$ since we have equality to zero whenever $E(X)=E(Y)$ in (\ref{energy distance alpha}).

The two-sample energy statistic corresponding to energy distance $\mathcal{E}^{\alpha}(X,Y)$, for independent random samples $\mathbf{X}=X_1,..., X_{n_1}$ and $\mathbf{Y}=Y_1,..., Y_{n_2}$ is
\begin{eqnarray}
\nonumber \mathcal{E}_{n_1,n_2}^{\alpha}(\mathbf{X},\mathbf{Y})&=&\frac{2}{n_1n_2}\sum_{i=1}^{n_1}\sum_{m=1}^{n_2}|X_i-Y_m|^{\alpha}-\\
\nonumber&&\frac{1}{n_1^2}\sum_{i=1}^{n_1}\sum_{j=1}^{n_1}|X_i-X_j|^{\alpha}-\frac{1}{n_2^2}\sum_{l=1}^{n_2}\sum_{m=1}^{n_2}|Y_l-Y_m|^{\alpha},
\end{eqnarray}
where $\alpha \in (0,2)$. The weighted two-sample statistic
\begin{displaymath}
T_{X,Y}=\left(\frac{n_1n_2}{n_1+n_2}\right)\mathcal{E}_{n_1,n_2}(\mathbf{X},\mathbf{Y})
\end{displaymath}
determines a consistent test \cite{szekely2004testing} for equality of distributions of $X$ and $Y$. The multi-sample energy test for equality of $k$ distributions, $k \geq 2$, is given in \cite{rizzo2010disco}.

\section{K-groups}

$K$-means usually applies quadratic distance to compute the dissimilarity between the data object and the prespecified prototype, and minimizes the variance within the clusters. In this paper, we apply a weighted two-sample energy statistic $T_{X,Y}$ as the statistical function to measure the dissimilarity between the clusters, and modify the $k$-means algorithm given by Hartigan and Wong in 1979 \cite{hartigan1979algorithm} to implement our algorithm. Generally, our method belongs to the class of distribution-based algorithms. This kind of algorithm takes a cluster as a dense region of data objects that is surrounded by regions of low densities. They are often employed when the clusters are irregular or intertwined, or when noise and outliers are present. Since the energy distance measures the similarity between two sets rather than the similarity between the object and prototype, we name our method \emph{$k$-groups}.

We define dispersion between two sets $A,B$ as
\begin{displaymath}
G^{\alpha}(A,B)=\frac{1}{n_1n_2}\sum_{i=1}^{n_1}\sum_{m=1}^{n_2}|a_{i}-b_{m}|^{\alpha},
\end{displaymath}
where $0 <\alpha \le 2$, and $n_1,n_2$ are the sample sizes for sets $A,B$. Let $P= \{\pi_1,...,\pi_k\}$ be a partition of observations, where $k$ is the number of clusters, prespecified. We define the total dispersion of the observed response by
\begin{displaymath}
 T^{\alpha}(\pi_1,...\pi_k)=\frac{N}{2}G^{\alpha}(\cup_{i=1}^k\pi_i,\cup_{i=1}^k\pi_i),
\end{displaymath}
where N is the total number of observations. The within-groups dispersion is defined by
\begin{displaymath}
W^{\alpha}(\pi_1,...\pi_k)=\sum_{j=1}^{k}\frac{n_j}{2}G^{\alpha}(\pi_j,\pi_j),
\end{displaymath}
where $n_j$ is the sample size for cluster $\pi_j$. The between-sample dispersion is
\begin{displaymath}
B^{\alpha}(\pi_1,...\pi_k)=
\sum_{1 \leq i< j \leq k}\left\{\frac{n_i n_j}{2N}
\left(2G^{\alpha}(\pi_i,\pi_j)-G^{\alpha}(\pi_i,\pi_i)-G^{\alpha}(\pi_j,\pi_j)\right)\right\}.
\end{displaymath}
When $0 < \alpha \leq 2$ we have the decomposition
\begin{displaymath}
T^{\alpha}(\pi_1,...\pi_k)=W^{\alpha}(\pi_1,...\pi_k)+B^{\alpha}(\pi_1,...\pi_k),
\end{displaymath}
where both $W^{\alpha}(\pi_1,...\pi_k)$ and $B^{\alpha}(\pi_1,...\pi_k)$ are nonnegative (Rizzo and Sz{\'e}kely 2010) applied this decomposition of $T^\alpha$ into distance components (disco) to obtain a consistent nonparametric test for equality of $k$ distributions.

To maximize between-sample dispersion, $B^{\alpha}(\pi_1,...\pi_k)$, with $T^{\alpha}(\pi_1,...\pi_k)$ constant, is equivalent to minimizing $W^{\alpha}(\pi_1,...\pi_k)$. Hence, our purpose is to find the best partitions which minimize the within-cluster dispersion $W^{\alpha}$. Therefore, the objective function for $k$-groups is
\begin{equation} \label{diff of w}
     min_{\pi_1,...,\pi_k}\sum_{j=1}^{k}\frac{n_j}{2}G^{\alpha}(\pi_j,\pi_j)=min_{\pi_1,...,\pi_k}W^{\alpha}(\pi_1,...\pi_k).
\end{equation}

\subsection{First Variation Algorithm}
Motivated by Hartigan and Wong's idea, we searched for a $k$-partition with locally optimal $W^{\alpha}$ by moving points from one cluster to another. We call this reallocation step \emph{First Variation}.
\begin{definition}
A first variation of a partition $P$ is a partition $P'$ obtained from $P$ by removing a single point $\textbf{a}$ from a cluster $\pi_{i}$ of $P$ and assigning this point to an existing cluster $\pi_j$ of $P$.
\end{definition}

Let $\pi_1$ and $\pi_2$ be two different clusters in partition $P= {\pi_1,...,\pi_k}$, and point $\textbf{a} \in \pi_i$. Cluster $\pi^{-}_1$ represents cluster $\pi_1$ after removing point $\textbf{a}$, and cluster $\pi^{+}_2$ represents cluster $\pi_2$ after adding point $\bf{a}$. Let $n_1$ and $n_2$ be the sizes of cluster $\pi_1$ and $\pi_2$. The within-cluster dispersion of $\pi_1$ and $\pi_2$ are
\begin{align}
\notag &G^{\alpha}(\pi_{1},\pi_{1})=\frac{1}{2n_{1}} \sum_{i}^{n_{1}} \sum_{j}^{n_{1}} |x_{i}^{1}-x_{j}^{1}|^{\alpha}, \\
\notag &G^{\alpha}(\pi_{2},\pi_{2})=\frac{1}{2n_{2}} \sum_{i}^{n_{2}} \sum_{j}^{n_{2}} |x_{i}^{2}-x_{j}^{2}|^{\alpha},
\end{align}
where $x_{i}^{1} \in \pi_1, i=1,...,n_1$ and $x_{i}^{2}\in\pi_2, i=1,...,n_2$. The within-cluster dispersion of $\pi^{-}_1$ and $\pi^{+}_2$ are
\begin{align}
\notag &G^{\alpha}(\pi^{-}_{1},\pi^{-}_{1})=\frac{1}{2\cdot(n_{1}-1)} \sum_{i}^{n_{1}-1} \sum_{j}^{n_{1}-1} |x_{i}^{-1}-x_{j}^{-1}|^{\alpha}, \\
\notag &G^{\alpha}(\pi^{+}_{2},\pi^{+}_{2})=\frac{1}{2\cdot(n_{2}+1)} \sum_{i}^{n_{2}+1} \sum_{j}^{n_{2}+1} |x_{i}^{+2}-x_{j}^{+2}|^{\alpha},
\end{align}
where $x_{i}^{-1} \in \pi^{-}_{1}, i=1,...,n_1-1$ and $x_{i}^{+2}\in\pi^{+}_2, i=1,...,n_2+1$. The two-sample energy statistics between point $\textbf{a}$ with cluster $\pi_1$ and $\pi_2$ are
\begin{equation} \label{epi1}
\xi^{\alpha}( \textbf{a},\pi_{1})=\frac{2}{n_{1}}\sum_{i}^{n_{1}}|x_{i}^{1}-\textbf{a}|^{\alpha}-
\frac{1}{n_1^2} \sum_{i}^{n_{1}} \sum_{j}^{n_1} |x_{i}^{1}-x_{j}^{1}|^{\alpha},
\end{equation}

\begin{equation} \label{epi2}
\xi^{\alpha}(\textbf{a},\pi_{2})=\frac{2}{n_{2}}\sum_{i}^{n_2}|x_{i}^{2}-\textbf{a}|^{\alpha}-
\frac{1}{n_{2}^{2}} \sum_{i}^{n_{2}} \sum_{j}^{n_2} |x_{i}^{2}-x_{j}^{2}|^{\alpha}.
\end{equation}

The resulting change $W^{\alpha}(P)-W^{\alpha}(P')$ in the within component is derived as follows.

\begin{enumerate}[(i)]
\item
 First, we compute $G^{\alpha}(\pi_{1},\pi_{1})-G^{\alpha}(\pi_{1}^{-},\pi_{1}^{-})$, as
\begin{align}   % \nonumber to remove numbering (before each equation)
  \notag G^{\alpha}&(\pi_{1},\pi_{1})-G^{\alpha}(\pi_{1}^{-},\pi_{1}^{-})\\
  \notag          &= \frac{1}{2\cdot n_{1}}\sum_{i}^{n_{1}}\sum_{j}^{n_{1}} |x_{i}^{1}-x_{j}^{1}|^{\alpha}-\frac{1}{2\cdot(n_{1}-1)}\sum_{i}^{n_{1}-1}\sum_{j}^{n_{1}-k} |x_{i}^{-1}-x_{j}^{-1}|^{\alpha}  \\
  \notag  &=\frac{1}{2\cdot n_{1}}\sum_{i}^{n_{1}}\sum_{j}^{n_{1}}|x_{i}^{1}-x_{j}^{1}|^{\alpha}-\frac{1}{2\cdot(n_{1}-1)}
                 \{ \sum_{i}^{n_{1}}\sum_{j}^{n_{1}}|x_{i}^{1}-x_{j}^{1}|^{\alpha}\\
  \notag  &-2\sum_{i}^{n_{1}}|x_{i}^{1}-\textbf{a}|^{\alpha}
                 \}\\
  \label{diff Gpi1} =&\frac{1}{n_{1}-1}\sum_{i}^{n_{1}}|x_{i}^{1}-\textbf{a}|^{\alpha}-\frac{1}{2\cdot n_{1}(n_{1}-1)}\sum_{i}^{n_{1}}\sum_{j}^{n_{1}}|x_{i}^{1}-x_{j}^{1}|^{\alpha}.
 \end{align}
\item Multiply $\frac{n_{1}}{2(n_{1}-1)}$ times equation (\ref{epi1}) to obtain
 \begin{align}
\notag \frac{n_{1}}{2(n_{1}-1)}\xi^{\alpha}(\textbf{a},\pi_{1})&=\frac{1}{n_{1}-1}\sum_{i}^{n_{1}}|x_{i}^{1}-a|^{\alpha}\\
 \label{weighted epi1}  &-\frac{1}{2\cdot n_{1}(n_{1}-1)}\sum_{i}^{n_{1}}\sum_{j}^{n_{1}}|x_{i}^{1}-x_{j}^{1}|^{\alpha}.
 \end{align}
\item Subtract (\ref{weighted epi1}) from (\ref{diff Gpi1}):
\begin{equation} \label{Gpi1 weighted epil}
G^{\alpha}(\pi_{1},\pi_{1})-G^{\alpha}(\pi_{1}^{-},\pi_{1}^{-})=\frac{n_{1}}{2(n_{1}-1)}\xi^{\alpha}(\textbf{a},\pi_{1}).
\end{equation}
\item Then compute $G^{\alpha}(\pi_{2}^{+},\pi_{2}^{+})-G^{\alpha}(\pi_{2},\pi_{2})$ as
\begin{align}   % \nonumber to remove numbering (before each equation)
  \notag &G^{\alpha}(\pi_{2}^{+},\pi_{2}^{+})-G^{\alpha}(\pi_{2},\pi_{2}) \\
  \notag &= \frac{1}{2\cdot (n_{2}+1)}\sum_{i}^{n_{2}+1}\sum_{j}^{n_{2}+1} |x_{i}^{+2}-x_{j}^{+2}|^{\alpha}-\frac{1}{2 \cdot
                 n_{2}}\sum_{i}^{n_{2}}\sum_{j}^{n_{2}} |x_{i}^{2}-x_{j}^{2}|^{\alpha}  \\
  \notag  &= \frac{1}{2\cdot (n_{2}+1)} \{\sum_{i}^{n_{2}}\sum_{j}^{n_{2}}|x_{i}^{2}-x_{j}^{2}|^{\alpha}+2\sum_{i}^{n_{2}}|x_{i}^{2}-\bf{a}|^{\alpha}\}\\
  \notag  &-\frac{1}{2 \cdot n_{2}}\sum_{i}^{n_{2}}\sum_{j}^{n_{2}} |x_{i}^{2}-x_{j}^{2}|^{\alpha} \\
 \label{diff Gpi2} &= \frac{1}{n_{2}+1}\sum_{i}^{n_{2}}|x_{i}^{2}-a|^{\alpha}-\frac{1}{2\cdot n_{2}(n_{1}+1)}\sum_{i}^{n_{2}}\sum_{j}^{n_{2}}|x_{i}^{2}-x_{j}^{2}|^{\alpha}.
 \end{align}
\item Multiply $\frac{n_{2}}{2(n_{1}+1)}$ times equation (\ref{epi2}):
\begin{align}
\notag \frac{n_{2}}{2(n_{1}+1)}\xi^{\alpha}(\textbf{a},\pi_{2})&=\frac{1}{n_{2}+1}\sum_{i}^{n_{2}}|x_{i}^{2}-\textbf{a}|^{\alpha}\\
 \label{weighted epi2}    &-\frac{1}{2\cdot n_{2}(n_{2}+1)}\sum_{i}^{n_{2}}\sum_{j}^{n_{2}}|x_{i}^{2}-x_{j}^{2}|^{\alpha}.
 \end{align}
\item Subtract (\ref{weighted epi2}) from (\ref{diff Gpi2}) to obtain
\begin{equation}\label{Gpi2 weighted epi2}
G^{\alpha}(\pi_{2}^{+},\pi_{2}^{+})-G^{\alpha}(\pi_{2},\pi_{2})=\frac{n_{2}}{2(n_{2}+1)}\xi^{\alpha}(\textbf{a},\pi_{2}).
\end{equation}
\item Finally, subtract (\ref{Gpi2 weighted epi2}) from (\ref{Gpi1 weighted epil}), which gives
\begin{align}
\notag G^{\alpha}&(\pi_{1},\pi_{1})+G^{\alpha}(\pi_{2},\pi_{2})- G^{\alpha}(\pi^{-}_{1},\pi^{-}_{1})-G^{\alpha}(\pi^{+}_{2},\pi^{+}_{2})\\
\label{first variation}       &=\frac{n_{1}}{2(n_{1}-1)}\xi^{\alpha}(\textbf{a},\pi_{1})-\frac{n_{2}}{2(n_{2}+1)}\xi^{\alpha}(\textbf{a},\pi_{2}).
\end{align}
\end{enumerate}

Based on the derivation above, we have established the following theorem.

\begin{theorem}
Suppose that $P =\{\pi_1,\pi_2,...\pi_{k}\}$ is a partition of the data, and $P^a=\{\pi_{1}^{-},$ $\pi_{2}^{+},...,\pi_{k}\}$ is the partition
obtained by moving point $\mathbf a$ from $\pi_1$ to $\pi_2$. Then
\begin{equation}\label{Thm move one point}
W^{\alpha}(P)-W^{\alpha}(P')=\frac{n_{1}}{2(n_{1}-1)}\xi^{\alpha}(\textbf{a},\pi_{1}) -\frac{n_{2}}{2(n_{2}+1)}\xi^{\alpha}(\textbf{a},\pi_{2}).
\end{equation}
\end{theorem}
Similar to the Hartigan and Wong $k$-means algorithm, the $k$-groups algorithm moves point $\bf{a}$ from cluster $\pi_1$ to $\pi_2$ if
$$
\frac{n_{1}}{2(n_{1}-1)}\xi^{\alpha}(\textbf{a},\pi_{1}) -\frac{n_{2}}{2(n_{2}+1)}\xi^{\alpha}(\textbf{a},\pi_{2})
$$ is positive. Otherwise point $\bf{a}$ remains in cluster $\pi_1$. Based on the computation above, we propose the following $k$-groups algorithm.

\par \textbf{Notation} Let $N$ be the total sample size of observations, $M$ be the dimension of the sample, and $K$ be the prespecified number of clusters. The number of points in cluster $\pi_i$ is denoted by $n_i,$ $i=1,...,K)$. The two-sample energy statistic between point $I$ to cluster $\pi_i$ is denoted by $\xi^{\alpha}(I,\pi_i)$. The $k$-groups algorithm is the following.
\subsection*{K-groups Clustering Algorithm}
\begin{itemize}
 \item[Step 1] For each point $I$, $I= 1,...,N$, randomly assign $I$ to cluster $\pi_i, i=1,...,K$. Let $\pi(I)$ represent the cluster where $I$ belongs, and $n(\pi(I))$ represent the size of cluster $\pi(I)$.
 \item[Step 2] For each point $I,$ $I= 1,...,N)$, compute
 $$
 E_1=\frac{n(\pi(I))}{2(n(\pi(I))-1)}\xi^{\alpha}(I,\pi(I))
 $$ and
 $$
 E_2=min \left[\frac{n(\pi_i)}{2(n(\pi_i)+1)}\xi^{\alpha}(I,\pi_i)\right]
 $$
 for all clusters $\pi_i$, $\pi\ne \pi(I)$. If $E_1$ is less than $E_2$, observation $I$ remains in cluster $\pi(I)$. Otherwise, move the point $I$ to cluster $\pi$, and update clusters $\pi(I)$ and $\pi$.
 \item[Step 3] Stop if there is no relocation in the last $N$ steps.
\end{itemize}

\subsection{K-means as a Special Case of K-groups}

In this section we show that $k$-means is a special case of $k$-groups when $\alpha=2$. According to the properties of energy distance, we know that when $0<\alpha<2$, the energy distance $\xi^{\alpha}(X,Y)=0$ if and only if random variables $X$ and $Y$ follow the same statistical distribution. However, when $\alpha =2$, we have $\xi^{\alpha}(X,Y)=0$ whenever $E(X)=E(Y)$.

Theorem \ref{k-g k-m Thm 1} shows that $k$-groups contains the $k$-means  algorithm as a special case when $\alpha = 2$, by showing that Hartigan and Wong's $k$-means algorithm has the same objective function as $k$-groups when $\alpha=2$.

\begin{proposition} \label{k-g k-m p1}
For the $k$-groups algorithm (first variation) with exponent $\alpha=2$
 \begin{equation*}
  \frac{n_i}{2}G^{\alpha}(\pi_i,\pi_i)=\sum^{n_i}_{l=1}x_l^2-n_ic_i^2,
 \end{equation*}
where $c_{i}=\frac{1}{n_i}\sum^{n_i}_{j=1}x_j$,and $x_j\in \pi_i$, $j=1,...n_i$
\end{proposition}
\begin{proof}
\begin{align}
\notag \frac{n_i}{2}G^{2}(\pi_i,\pi_i)&=\frac{1}{2n_i}\sum_{l=1}^{n_i}\sum_{m=1}^{n_i}|x_l-x_m|^2\\
\notag                            &=\frac{1}{2n_i}\sum_{l=1}^{n_i}\sum_{m=1}^{n_i}(x_l^2-2x_lx_m+x_m^2)\\
\notag                            &=\frac{1}{2n_i}[n_i\sum_{l=1}^{n_i}x_l^2-2\sum_{l=1}^{n_i}\sum_{m=1}^{n_i}x_lx_m+n_i\sum_{m=1}^{n_i}x_m^2]\\
\notag                            &=\frac{1}{2n_i}[2n_i\sum_{l=1}^{n_i}x_l^2-2\sum_{l=1}^{n_i}\sum_{m=1}^{n_i}x_lx_m]\\
\notag                            &=\frac{1}{2n_i}[2n_i\sum_{l=1}^{n_i}x_l^2-2n_i^2c_i^2]\\
\label{k-groups}                  &=\sum^{n_i}_{l=1}x_l^2-n_ic_i^2.
\end{align}
\end{proof}

\begin{theorem} \label{k-g k-m Thm 1}
When $\alpha=2$, the $k$-groups algorithm and the Hartigan and Wong $k$-means algorithm have the same objective function.
\end{theorem}
\begin{proof}
%he objective function for $k$-means is
\begin{displaymath}
 min_{c_{i}, 1\le i\le k}\sum_{i=1}^k\sum_{x_j\in \pi_{i}} (x_j-c_{i})^2,
\end{displaymath}
and the objective function for $k$-groups is
\begin{displaymath}
min_{\pi_1,...,\pi_k}\sum_{i=1}^{k}\frac{n_i}{2}G^{\alpha}(\pi_i,\pi_i).
\end{displaymath}
By Proposition \ref{k-g k-m p1}
 $$
 \sum_{x_j\in \pi_{i}}(x_j-c_{i})^2=\frac{n_i}{2}G^{2}(\pi_i,\pi_i),
 $$
 for all $i=1,...,k$. Hence, when $\alpha =2$, $k$-groups and $k$-means have the same objective function.
\end{proof}

Next we show that the update formula of $k$-groups and Hartigan and Wong's $k$-means algorithm are the same when $\alpha=2$.

\begin{proposition} \label{k-g k-m p3}
Suppose that a point $I$ belongs to cluster $L$, and the sample size of $L$ is $n$. Then
\begin{equation}\label{update1}
\frac{n}{2(n-1)}\xi^{2}(I,L)=\frac{n\cdot D(I,L)^2}{n-1},
\end{equation}
and
\begin{equation} \label{update2}
\frac{n}{2(n+1)}\xi^{2}(I,L)=\frac{n \cdot D(I,L)^2}{n+1},
\end{equation}
where
\begin{equation}\label{kmeans update formula}
D(I,L)^2=\left(I-\frac{\sum_{i=1}^n x_i}{n}\right)^2.
\end{equation}
is the $k$-means updating formula.

Hence in the special case $\alpha=2$, $k$-groups and $k$-means
algorithms apply the same updating formula.
\end{proposition}
\begin{proof}
We only need to prove that $\frac{1}{2}\xi^2(I,L)=D(I,L)^2$. Setting  $\alpha=2$, we have
\begin{equation} \label{energy update formula}
\xi^2(I,L)=\frac{2}{n}\sum_{i=1}^n|I-x_i|^2-\frac{1}{n^2}\sum_{i=1}^{n}\sum_{j=1}^{n}|x_i-x_j|^2.
\end{equation}
We can simplify equation (\ref{energy update formula}) as follows:
\begin{align}
\notag \xi^2(I,L)&=\frac{2}{n}\sum_{i=1}^n|I-x_i|^2-\frac{1}{n^2}\sum_{i=1}^{n}\sum_{j=1}^{n}|x_i-x_j|^2\\
\notag           &=\frac{2}{n}\left(nI^2-2I\sum_{i=1}^n x_i + \sum_{i=1}^{n} x_i^2\right)-\frac{1}{n^2}\sum_{i=1}^n\sum_{j=1}^n(x_i^2-2x_ix_j+x_j^2)\\
\notag           &=\left(2I^2-4I\frac{\sum_{i=1}^n x_i^2}{n}+2\frac{\sum_{i=1}^n x_i^2}{n}\right)-\frac{1}{n^2}(2n\sum_{i=1}^n x_i^2-2\sum_{i=1}^n\sum_{j=1}^nx_ix_j)\\
\notag           &=\left(2I^2-4I\frac{\sum_{i=1}^n x_i^2}{n}+2\frac{\sum_{i=1}^n x_i^2}{n}\right)-(2\frac{\sum_{i=1}^n x_i^2}{n}-2\frac{2\sum_{i=1}^n\sum_{j=1}^nx_ix_j}{n^2})\\
\notag           &=2\left(I^2-2I\frac{\sum_{i=1}^n x_i}{n}+\frac{2\sum_{i=1}^n\sum_{j=1}^nx_ix_j}{n^2}\right)\\
\notag           &=2\left(I-\frac{\sum_{i=1}^n x_i}{n}\right)^2.
\end{align}
Thus
\begin{equation*}
\frac{1}{2}\xi^2(I,L)=D(I,L)^2.
\end{equation*}
With similar steps we obtain (\ref{update2}).
\end{proof}

\subsection{Second Variation Algorithm}
The objective of $k$-groups is to find a global minimum of the within-cluster sum of dispersion. However, in most cases we can only find a local minimum by first variation method. Usually in order to solve this problem, one can try different initial random starts, and choose the best result with minimum within-cluster dispersion. In order to more closely achieve  global optimization, we consider modifying our algorithm to move more than one observation at each step. The reasons to move more than one point are the following.
\begin{itemize}
\item It allows the algorithm to move from a local optimum obtained by the first variation.
\item Based on the result of \cite{szekely2004testing}, if two samples follow different distributions, the weighted two-sample energy statistic
     $$
     T_{X,Y}=\left(\frac{n_1n_2}{n_1+n_2}\right)\mathcal{E}_{n_1,n_2}(\mathbf{X},\mathbf{Y})
     $$
     will approach infinity stochastically as $N=n_1+n_2$ tends to infinity and neither $\frac{n_1}{N}$ nor $\frac{n_2}{N}$ goes to zero.
\item Energy statistics admit a nice update formula for moving $m$ points. We will show that the difference of within-cluster sum of dispersion equals the difference of weighted two-sample energy statistics if we move any $m$ $(m>1)$ points from cluster to cluster.
\end{itemize}

\begin{definition}
An $m^{th}$ variation of a partition $P$ is a partition $P^{(m)}$ obtained from $P$ by removing $m$ points $\{a_1, a_2,..., a_m\}$ from a cluster $\pi_{i}$ of $P$ and assigning these points to an existing cluster $\pi_j$ of $P$, $i \ne j$.
\end{definition}

The $k$-groups second variation algorithm moves two points in each step from one cluster to another cluster. For the implementation, we need to derive the corresponding difference in within cluster dispersion, $W^\alpha(P) - W^\alpha(P^{(2)})$. Below, we derive the general result for moving $m$ $(m > 1)$ points. The formulas for the second variation are given in (\ref{update-mth}) and (\ref{m-th rule}) with $m=2$.

\begin{theorem}\label{mth variation th}
Suppose $P=\{\pi_1,\pi_2,...\pi_k\}$ is a partition, and $P^{(m)}=\{\pi_{1}^{-},\pi_{2}^{+},$ $..., \pi_{k}\}$ is a $m^{th}$ variation of $P$ by moving points $\{a_1, a_2,...,a_m\}$ from cluster $\pi_1$ to $\pi_2$. Then
\begin{align}\notag
W^{\alpha}(P)-W^{\alpha}(P^{(m)})&=\frac{mn_{1}}{2(n_{1}-m)}\xi^{\alpha}(\{a_1,...a_m\},\pi_{1})
\\ & \qquad -\frac{mn_{2}}{2(n_{2}+m)}\xi^{\alpha}(\{a_1,...a_m\},\pi_{2}). \label{update-mth}
\end{align}
\end{theorem}

\begin{proof}
We want to move $m$ points $\{a_1, a_2,...,a_m\}$ from cluster $\pi_{1}$ to another cluster $\pi_{2}$. Cluster $\pi_1^{-}$ represents cluster $\pi_1$ after removing $m$ points $\{a_1, a_2,...,a_m\}$, and cluster $\pi_2^{+}$ represents cluster $\pi_2$ after adding those $m$ points. Let $n_1$ and $n_2$ be the sizes of $\pi_1$ and $\pi_2$ before moving $m$ points. The two-sample energy statistic between the $m$ points $\{a_1, a_2,...,a_m\}$ and clusters $\pi_{1}$, $\pi_{2}$ are by definition
\begin{align}
\notag \xi^{\alpha}(\{a_1,...a_m\},\pi_{1})=&\frac{2}{m\cdot n_{1}}\sum_{i}^{n_1} \sum_{j}^{m}|x_{i}^{1}-a_j|^{\alpha}-\frac{1}{m^2}\sum_{i}^{m}\sum_{j}^{m}|a_{i}-a_{j}|^{\alpha}\\
\label{epi1 m}                      & \qquad
-\frac{1}{n_1^2}\sum_{i}^{n_{1}}\sum_{j}^{n_{1}}|x_i^{1}-x_j^{1}|^{\alpha},
\end{align}
and
\begin{align}
\notag \xi^{\alpha}(\{a_1,...a_m\},\pi_{2})=&\frac{2}{m\cdot n_{2}}\sum_{i}^{n_{2}} \sum_{j}^{m}|x_{i}^2-a_j|^{\alpha}-\frac{1}{m^{2}}\sum_{i}^{m}\sum_{j}^{m}|a_{i}-a_{j}|^{\alpha}\\ & \qquad
\label{epi2 m }
-\frac{1}{n_{2}^{2}} \sum_{i}^{n_{2}} \sum_{j}^{n_{2}} |x_{i}^{2}-x_{j}^{2}|^{\alpha}.
\end{align}

The $m$-th variation updating formula is derived as follows.
Similar to the derivation of first variation, we compute $\frac{n_1}{2}G^{\alpha}(\pi_{1},\pi_{1})-\frac{n_1-m}{2}G^{\alpha}(\pi_{1}^{-},\pi_{1}^{-})$ and $\frac{n_2+m}{2}G^{\alpha}(\pi_{2}^{+},\pi_{2}^{+})-\frac{n_2}{2}G^{\alpha}(\pi_{2},\pi_{2})$, as
\begin{align}
  \notag \frac{n_1}{2}  G^{\alpha}  (\pi_{1},\pi_{1}) &   -\frac{n_1-m}{2}G^{\alpha}(\pi_{1}^{-},\pi_{1}^{-})
  \\ \notag = &
  \frac{1}{2 n_{1}}\sum_{i}^{n_{1}}\sum_{j}^{n_{1}} |x_{i}^{1}-x_{j}^{1}|^{\alpha}-\frac{1}{2 \cdot
                 (n_{1}-m)}\sum_{i}^{n_{1}-m}\sum_{j}^{n_{1}-m} |x_{i}^{-1}-x_{j}^{-1}|^{\alpha}
 \\   \notag  = &
 \frac{1}{2 n_{1}}\sum_{i}^{n_{1}}\sum_{j}^{n_{1}}|x_{i}^{1}-x_{j}^{1}|^{\alpha}  -\frac{1}{2(n_{1}-m)} \times
 \notag \\ &
         \left\{ \sum_{i}^{n_{1}}\sum_{j}^{n_{1}}|x_{i}^{1}-x_{j}^{1}|^{\alpha}
  -2\sum_{i}^{n_1}\sum_{j}^{m}|x_{i}^{1}-a_j|^{\alpha}+\sum_{i}^{m} \sum_{j}^{m}|a_{i}-a_{j}|^{\alpha}\right\}  \notag
  \\   \notag & =\frac{1}{n_{1}-m}  \sum_{i}^{m}\sum_{j}^{n_{1}}|x_{i}^1-a_{j}|^{\alpha}-\frac{1}{2(n_{1}-m)}
                 \sum_{i}^{m}\sum_{j}^{m}|a_{i}-a_{j}|^{\alpha}
\\
 \label{diff Gpi1 m} & \qquad -\frac{m}{2 n_{1}(n_{1}-m)}\sum_{i}^{n_{1}}\sum_{j}^{n_{1}}|x_{i}^{1}-x_{j}^{1}|^{\alpha}.
 \end{align}
Multiply equation(\ref{epi1 m}) by $\frac{m\cdot n_{1}}{2(n_{1}-m)}$ to get
 \begin{align}
 \notag \frac{mn_{1}}{2(n_{1}-m)}\xi^{\alpha}(\{a_1,...a_m\},\pi_{1})=&\frac{1}{n_{1}-m}\sum_{i}^{n_1}\sum_{j}^{m}|x_{i}^{1}-a_{j}|^{\alpha}-\\
 \notag &\frac{n_{1}}{2m(n_{1}-m)}\sum_{i}^{m}\sum_{j}^{m}|a_{i}-a_{j}|^{\alpha}-\\
 \label{weight epi1 m}  &\frac{m}{2n_{1}(n_{1}-m)}\sum_{i}^{n_{1}}\sum_{j}^{n_{1}}|x_{i}^{1}-x_{j}^{1}|^{\alpha}.
 \end{align}
Then subtract equation (\ref{weight epi1 m}) from equation (\ref{diff Gpi1 m}) to obtain
\begin{align}
\notag                        \frac{n_1}{2}G^{\alpha}(\pi_{1},\pi_{1})&-\frac{n_1-m}{2}G^{\alpha}(\pi_{1}^{-},\pi_{1}^{-})\\
\label{Gpi1 weight epi1 m}    =&\frac{mn_{1}}{2(n_{1}-m)}\xi^{\alpha}(\{a_1,...a_m\},\pi_{1})+\frac{1}{2m}\sum_{i}^{m}\sum_{j}^{m}|a_{i}-a_{j}|^{\alpha}.
\end{align}
Following similar steps used in deriving the first variation formula, we have
\begin{align}
\notag \frac{n_2+m}{2}G^{\alpha}(\pi_{2}^{+},\pi_{2}^{+})&-\frac{n_2}{2}G^{\alpha}(\pi_{2},\pi_{2})\\
\label{Gpi2 weight epi2 m}=&\frac{mn_{2}}{2(n_{2}+m)}\xi^{\alpha} (\{a_1,...a_m\},\pi_{2})+\frac{1}{2m}\sum_{i}^{m}\sum_{j}^{m}|a_{i}-a_{j}|^{\alpha}.
\end{align}
Then subtracting equation (\ref{Gpi2 weight epi2 m}) from equation (\ref{Gpi1 weight epi1 m}) gives
\begin{align}
\notag \frac{n_1}{2}G^{\alpha}(\pi_{1},\pi_{1})&+\frac{n_2}{2}G^{\alpha}(\pi_{2},\pi_{2})-\frac{n_1-m}{2}G^{\alpha}(\pi_{1}^{-},\pi_{1}^{-})-\frac{n_2+m}{2}G^{\alpha}(\pi_{2}^{+},\pi_{2}^{+})\\
\label{m variation} =&\frac{mn_{1}}{2(n_{1}-m)}\xi^{\alpha}(\{a_1,...a_m\},\pi_{1})-\frac{mn_{2}}{2(n_{2}+m)} \xi^{\alpha}(\{a_1,...a_m\},\pi_{2}).
\end{align}
\end{proof}

Similar to first variation, we assign points $\{a_1, a_2,...,a_m\}$ to cluster $\pi_2$ if
\begin{equation}\label{m-th rule}
\frac{mn_{1}}{2(n_{1}-m)}\xi^{\alpha}(\{a_1,...a_m\},\pi_{1})-\frac{mn_{2}}{2(n_{2}+m)} \xi^{\alpha}(\{a_1,...a_m\},\pi_{2})
\end{equation}
is positive; otherwise we keep points $\{a_1, a_2,...,a_m\}$ in cluster $\pi_1$.

By Theorem \ref{mth variation th}, we have the objective function and updating formula for $m$-th variation $k$-groups clustering algorithms, $m > 1$. The second variation $k$-groups algorithm is the special case $m=2$.  However, the computation time to move more points is excessive. Suppose the total sample size is $N$, and we have $K$ clusters, with $K$ prespecified. For $k$-groups by first variation algorithm, we compute distance $NK$ times in each loop. If $m=2$, we compute distance $\frac{KN(N-1)}{2}$ times in each loop, because there are $\frac{N(N-1)}{2}$ combinations of two points. For the $m^{th}$ variation, we compute distance $\frac{CN!}{m!(N-m)!}$ times in each loop. The computation time will increase exponentially in $N$. Even though we have formulas for moving $m$ points, in practice the computation time is excessive for larger $m$. Here we let $m=2$ and implement the second variation algorithm by moving two points.

It is not practical to consider all possible combinations of two points. In fact, since our objective is to minimize the within-group sum of dispersion, we do not need to consider all possible combinations. We pair two points together if they have minimum distance, and we assume these two points should be assigned to the same cluster.

\textbf{Notation} Let even number $N$ be the total sample size of observations, $M$ be the dimension of the sample, and $K$ be the number of clusters, $K$ prespecified. The size of cluster $\pi_i,\, ( i=1,...,K)$ is denoted by $n_i$. The two-sample energy statistic between pair $II$ and cluster $\pi_i$ is denoted by $\xi^{\alpha}(II,\pi_i)$.

\subsection{K-groups Clustering Algorithm by Second Variation}

\begin{itemize}
\item[Step 1] Each pair of points $II$, $II=1,...,N/2$, is randomly assigned to a cluster $\pi_i\,(i=1,...,K)$. Let $\pi(II)$ represent the cluster containing pair $II$, and $n(\pi(II))$ represent the size of cluster $\pi(II)$.
\item[Step 2] For each pair $II\,(II=1,...,N/2)$, compute
$$
    E_1=\frac{n(\pi(II))}{n(\pi(II))-2}\xi^{\alpha}(II,\pi(II))
$$
and
$$
    E_2=\mathrm{min}\left[ \frac{n(\pi_i)}{n(\pi_i)+2}\xi^{\alpha}(II,\pi_i) \right]
$$
for all clusters $\pi_i$, where $\pi_i \neq \pi(II)$. If $E_1$ is less than $E_2$, pair $II$ remains in cluster $\pi(II)$; otherwise, move the pair $II$ to cluster $\pi_i$ with minimum value of $E_2$, and update cluster $\pi(II)$ and $\pi_i$.
\item[Step 3] Stop if there is no relocation in the last $\frac{N}{2}$ steps.
\end{itemize}

For an odd number $N$, we randomly take one point out. After applying $k$-groups by second variation, we assign that point to the cluster based on the updated formula of $k$-groups by first variation algorithm.

\section{Simulation Results}

A variety of cluster structures can be generated as mixtures of different distributions. Each of our simulated data sets was designed as a mixture, where each component of the mixture corresponds to a cluster. Each mixture distribution is simulated at a sample size $200$. We calculated average and standard error for validation indices diagonal (Diag), Kappa, Rand, and corrected Rand (cRand) based on $B=500$ iterations. In $k$-groups methods, for the mixture distributions which have finite first and second moments, we use $\alpha=1$; otherwise we use the smaller value of $\alpha=0.5$ to have finite moments $E|X-Y|^{\alpha}$. All algorithms were implemented in R \citep*{rcitation} and all simulations carried out in R.  Implemented $k$-groups algorithms are available upon request in an R package \emph{kgroups} \citep*{kgroups}. We want to compare $k$-groups by first variation, $k$-groups by second variation and $k$-means under different cluster structures.

\subsubsection*{Overlapping clusters}

Overlapping clusters are generated by location mixtures
with varying centers, so that the degree of separation of
clusters varies according to the distance between centers.
In these examples, a two component mixture was generated and
all algorithms set $k=2$.

Figure \ref{overlapping effect normal} displays the simulation results for 50\% normal location mixtures: \\
$0.5\,\mathrm{N}\,(0,1)+0.5\,\mathrm{N}\,(d,1)$, where $d=0.2,0.4,...,3$. The average cRand indices of the three algorithms are almost the same for each value of $d$. The results for symmetric normal mixtures suggest that both $k$-groups algorithms and $k$-means have similar performance when the clusters are normally distributed.

\begin{figure}[ht]
\begin{center}
\includegraphics[width=0.6\linewidth]{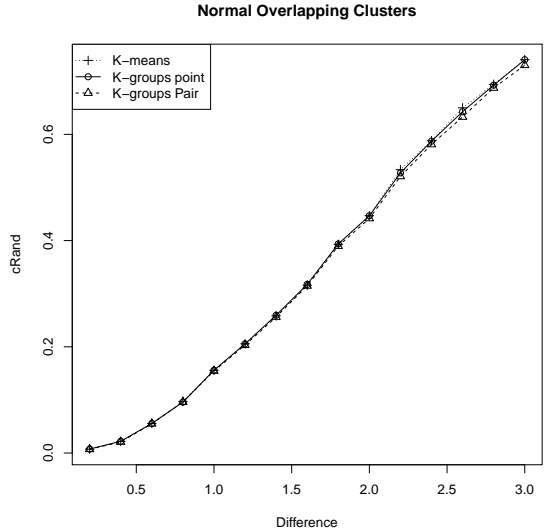}
\end{center}
\caption{Overlapping clusters effect for normal mixture distributions, $n=200,\, B=500$}
\label{overlapping effect normal}
\end{figure}

Figure \ref{overlapping effect lognormal} displays the simulation results for lognormal mixtures $0.5\,\mathrm{lognormal}(0,1)$ $+0.5\,\mathrm{lognormal}\,(d,1)$,\, where $d=0.5,1,...,10$. The average cRand indices of both $k$-groups algorithms dominate the $k$-means for each value of $d$. Thus, the results suggest that the $k$-groups algorithms have much better performance than $k$-means when clusters are strongly skewed, heavy tailed.

% overlapping effect of lognormal
\begin{figure}[ht]
\begin{center}
\includegraphics[width=0.6\linewidth]{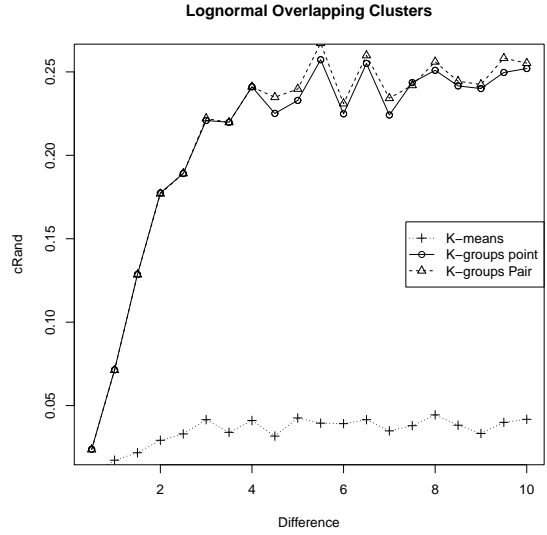}
\end{center}
\caption{Overlapping clusters effect for lognormal mixture distributions, $n=200,\, B=500$}
\label{overlapping effect lognormal}
\end{figure}

A Cauchy distribution does not have finite expectation, so neither
the $k$-means nor the $k$-groups ($\alpha \geq 1$) are valid.
However, the $k$-groups class of algorithms can apply a smaller
exponent $\alpha \in (0,1)$ on Euclidean distance such that
$E|X|^\alpha < \infty$.

Figure  \ref{overlapping effect cauchy} displays the simulation results for Cauchy mixtures  $0.5\,\mathrm{Cauchy}(0,1)+0.5\,\mathrm{Cauchy}$ $\,(d,1)$,\,where $d=0.2,0.4,...,3$, $\alpha=0.5$. The average cRand indices of both $k$-groups algorithms dominate the average cRand of $k$-means for each value of $d$.  Thus, the results suggest that both $k$-groups algorithms are more robust with respect to outliers and heavy tails.

\begin{figure}[ht]
\begin{center}
\includegraphics[width=0.6\linewidth]{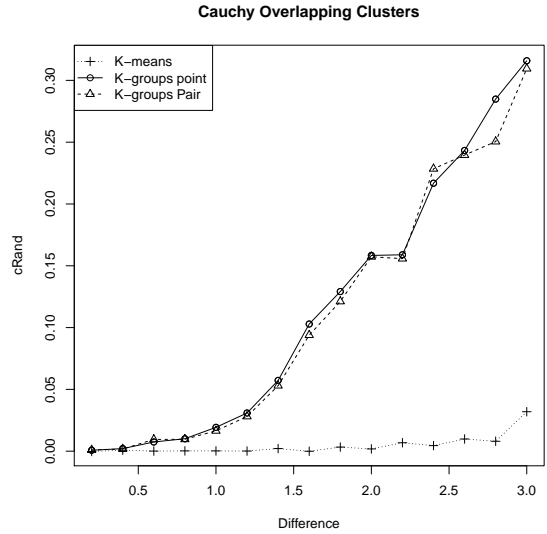}
\end{center}
\caption{Overlapping clusters effect for Cauchy mixture distributions, $n=200,\, B=500$}
\label{overlapping effect cauchy}
\end{figure}

\subsubsection*{Varying exponent on distance for k-groups}

In the following examples, the simulation design fixes the sample
size at $n=100$ while varying the exponent $\alpha$, $0<\alpha\leq 2$ on Euclidean distance for the $k$-groups methods. The $k$-means algorithm fixes $\alpha=2$, as shown in Theorem \ref{k-g k-m Thm 1}
and Proposition \ref{k-g k-m p1}.

Figure \ref{fig alpha effect normal} shows the results for normal mixtures
$0.5\,\mathrm{N}\,(0,1)+0.5\,\mathrm{N}\,(3,1)$\, with $\alpha= 0.2,0.4,...,2$. The average cRand indices of $k$-means and $k$-groups by first variation are very close. When $d=2$, $k$-means and $k$-groups by first variation have very close average cRand indices. The average cRand indices of $k$-groups by second variation are slightly lower than the other two algorithms. Generally, for each value of $\alpha$, the average cRand indices of both $k$-groups algorithms and $k$-means are very close. Thus, the results suggest that there is no $\alpha$ effect when clusters are normally distributed.

\begin{figure}[ht]
\begin{center}
\includegraphics[width=0.6\linewidth]{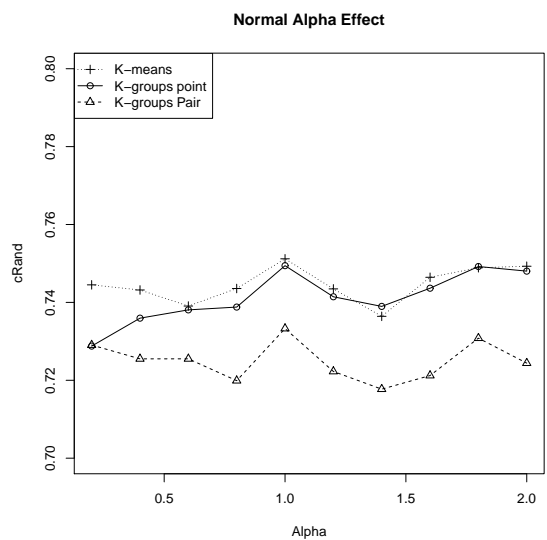}
\end{center}
\caption{Exponent $\alpha$ effect for normal mixture distributions, $n=200,\, B=1000$}
\label{fig alpha effect normal}
\end{figure}

Figure  \ref{fig alpha effect cauchy} shows results for Cauchy mixtures
$0.5$ Cauchy$(0,1)+0.5$ Cauchy $(3,1)$\, with varying exponent on distance $\alpha=0.2,0.4,...,2$. The average cRand indices of $k$-groups by first variation decrease as $\alpha$ increases, and when $\alpha=2$, the average cRand indices of $k$-groups by first variation and $k$-means are very close. In $k$-groups by second variation there are more stable average cRand indices than the other two algorithms. Thus, the results suggest that there is an $\alpha$ effect when clusters have an infinite first moment.

% alpha
\begin{figure}[h!t]
\begin{center}
\includegraphics[width=0.6\linewidth]{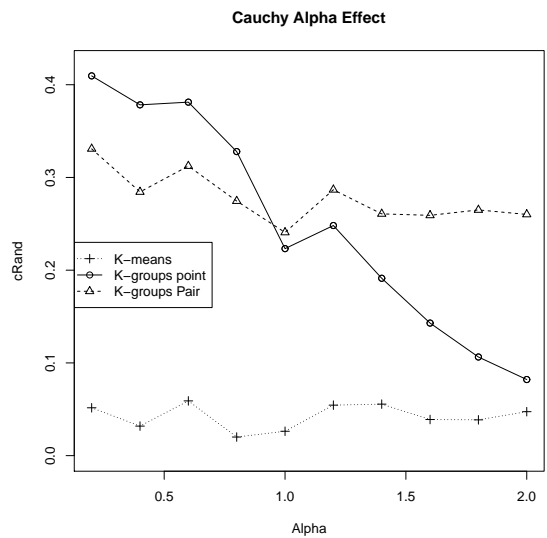}
\end{center}
\caption{Exponent $\alpha$ effect for Cauchy mixture distributions, $n=200,\, B=1000$}
\label{fig alpha effect cauchy}
\end{figure}

\subsubsection*{Effect of increasing dimension}

This set of simulations illustrates the effect of increasing
dimension on the performance of the clustering algorithms.

Table \ref{0.5 mv uniform (0,1) (0.3,0.7)} and Figure \ref{hduniform} summarize the simulation results for multivariate cubic mixtures  $0.5\,\mathrm{Cubic}^d\,(0,1)+0.5\,\mathrm{Cubic}^d\,(0.3,0.7)$. For each algorithm, the average Rand and cRand indices increase as the dimension $d$ increases.
The average Rand and cRand indices of these three algorithms are almost the same when $d<5$. However, the average Rand and cRand indices of both $k$-groups algorithms are consistently higher than $k$-means when $d>5$. Furthermore, the average Rand and cRand indices of $k$-groups by first variation approach $1$ as dimension $d$ increases. Thus, the results suggest that $k$-groups by first variation algorithm has better performance than the other two algorithms when clusters are cubic shaped in the multivariate case.

\begin{figure}[!ht]
\begin{center}
\includegraphics[width=0.6\linewidth]{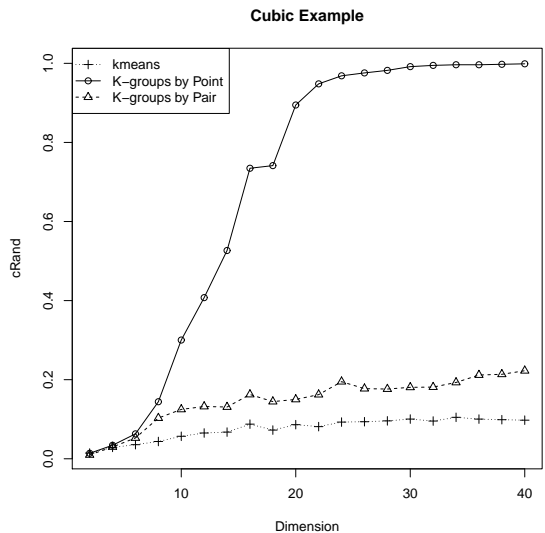}
\end{center}
\caption{Multivariate cubic mixtures, $d=2,4,...,40, \,n=200, \,B=500$}
\label{hduniform}
\end{figure}

% uniform
\begin{table}[!htbp]
\caption{Cubic Mixture, $\alpha=1$.}
\centering
$0.5\,\mathrm{Cubic}^d\,(0,1)+0.5\,\mathrm{Cubic}^d\,(0.3,0.7)$.
\begin{tabular} {l c c c c c}\\
\hline\hline
  Method          & d & Diag & Kappa & Rand & cRand \\
\hline
  $k$-means         &  1 &0.5381& 0.0758 & 0.5021 & 0.0043\\
  $k$-groups Point  &  1 &0.5352& 0.0710 & 0.5014 & 0.0028\\
  $k$-groups Pair   &  1 &0.5352& 0.0710 & 0.5014 & 0.0028\\

\hline
  $k$-means         &  2 &0.5439& 0.0877 & 0.5032 & 0.0065\\
  $k$-groups Point  &  2 &0.5440& 0.0879 & 0.5034 & 0.0068\\
  $k$-groups Pair   &  2 &0.5542& 0.0884 & 0.5034 & 0.0069\\
\hline
  $k$-means         &  5 &0.5536& 0.1067 & 0.5056 & 0.0113\\
  $k$-groups Point  &  5 &0.5713& 0.1427 & 0.5128 & 0.0257\\
  $k$-groups Pair   &  5 &0.5676& 0.1355 & 0.5120 & 0.0240\\
\hline
  $k$-means         &  10 &0.5705& 0.1393 & 0.5128& 0.0257\\
  $k$-groups Point  &  10 &0.7875& 0.5758 & 0.6923& 0.3847\\
  $k$-groups Pair   &  10 &0.6647& 0.3287 & 0.5672& 0.1346\\
\hline
  $k$-means          &  20 &0.6065 & 0.2078 & 0.5274 &0.0550\\
  $k$-groups Point   &  20 &0.9976 & 0.9951 & 0.9952 &0.9904\\
  $k$-groups Pair    &  20 &0.7213 & 0.4416 & 0.6045 &0.2090\\
\hline
   $k$-means         &  40 &0.6396 & 0.2794 & 0.5406 &0.0810\\
  $k$-groups Point   &  40 &0.9999 & 0.9999 & 0.9999 &0.9997\\
  $k$-groups Pair    &  40 &0.7471 & 0.4960 & 0.6228 &0.2456\\
\hline\hline
\end{tabular}
\label{0.5 mv uniform (0,1) (0.3,0.7)}
\end{table}

\section{Case Study}
\subsection{Diagnosis of Erythemato-Squamous Diseases in Dermatology}
\par The dermatology data analyzed is publicly available from the UCI Machine Learning Repository \citep{blake1998uci} at ftp.ics.uci.edu. The data was analyzed by \cite*{guvenir1998learning}, and contributed by G\"{u}venir. The erythemato-squamous dieases are proriasis, seboreic dermatitis, lichen planus, pityriasis rosea, choronic dermatitis and pityriasis rubra pilaris. According to \cite{guvenir1998learning}, diagnosis is difficult since all these diseases share the similar clinical features of erythema and scaling. Another difficulty is that a disease may show histopathological features of another disease initially, but have characteristic features at the following stages.

The data consists of $366$ objects with 34 attributes. There are $12$ clinical attributes and $22$ histopathological attributes. All except two take values in ${0,1,2,3}$, where $0$ indicates the feature was not present, and $3$ is the maximum. The attribute of family history takes value $0$ or $1$ and age of patient takes positive integer values. There are eight missing values in the age. The clinical and histopathological attributes are summarized in Table \ref{dermatology data summary}. We standardize all the attributes to zero mean and unit standard deviation and delete any observations which have missing values. The effective data size is 358 in the cluster analysis.

% dermatology example
\begin{table}[!ht]
\caption{Dermatology Data Summary}
\centering
\begin{tabular} {l l l l } \\
\hline
   & Clinical Attributes  &  & Histopathological Attributes \\[0.5ex]
\hline
\hline
1. & erythema                  & 12. & melanin incontinence \\
2. & scaling                   & 13. & eosinophils in the infiltrate \\
3. & definite borders          & 14. & PNL infiltrate                 \\
4. & itching                   & 15. & fibrosis of the paillary derims  \\
5. & koebner phenomenon        & 16. & exocytosis          \\
6. & polygonal papules         & 17. & acanthosis    \\
7. & follicular papules        & 18. & hyperkeratosis  \\
8. & oral mucosal involvement  & 19. & parakeratosis     \\
9. & knee and elbow involvement& 20. & clubbing of the rete ridges \\
10.& scalp involvement         & 21. & elongation of the rete ridges \\
11.& family history            & 22. & thinning of the suprapapillary epidermis \\
34.& age                       & 23. & pongiform pustule \\
   &                           & 24. & munro microabcess  \\
   &                           & 25. & focal hyperganulosis \\
   &                           & 26. & disapperance of the granular layer \\
   &                           & 27. & vaculolization and damage of basal layer \\
   &                           & 28. & spongiosis \\
   &                           & 29. & saw-tooth appearance of retes \\
   &                           & 30. & follicular horn plug \\
   &                           & 31. & perifollicular parakeratosis \\
   &                           & 32. & inflammatory mononuclear infiltrate \\
   &                           & 33. & band-like infiltrate \\
\hline
\end{tabular}
\label{dermatology data summary}
\end{table}
\par Table \ref{dermatology data results} shows the clustering result of $k$-means, $k$-groups by first variation, $k$-groups by second variation, and Hierarchical $\xi$. Hierarchical $\xi$ is agglomerative hierarchical clustering by energy distance; see  \cite{szekely2005hierarchical} for details. The maximum Rand and cRand index values $0.9740$ and $0.9188$ are obtained by $k$-groups by first variation. The Hierarchical $\xi$ obtains the second largest Rand and cRand index values $0.9730$ and $0.9159$. $k$-groups by second variation obtains the Rand and cRand index values $0.9543$ and $0.8602$. $k$-means obtains smallest Rand and cRand index values among those four algorithms: $0.9441$ and $0.8390$, respectively.

\begin{table}[!ht]
\caption{Dermatology Data Results}
\centering
\begin{tabular} {l c c c c}
\hline\hline
Indices & $k$-means & $k$-groups Point & $k$-groups Pair & Hierarchical $\xi$ \\[0.5ex]
\hline
 Diag  & $0.8324$  & $0.9553$  & $0.8910$ & $0.9497$ \\
 Kappa & $0.7882$  & $0.9440$  & $0.8640$ & $0.9370$ \\
 Rand  & $0.9441$  & $0.9740$  & $0.9543$ & $0.9730$ \\
 cRand & $0.8390$  & $0.9188$  & $0.8602$ & $0.9159$ \\
\hline\hline
\end{tabular}
\label{dermatology data results}
\end{table}

\section{Summary}

The simulation results for univariate and multivariate cases show that both $k$-groups algorithms perform as well as Hartigan and Wong's $k$-means algorithm when clusters are well-separated and normally distributed. Both $k$-groups algorithms perform better than $k$-means when data does not have a finite first moment. For data which has strong skewness and heavy tails, both $k$-groups algorithms perform better than $k$-means. For non--spherical clusters, both $k$-groups algorithms perform better than $k$-means in high dimensions and $k$-groups by first variation is consistent as dimension increases. Results of clustering on the dermatology data show that both $k$-groups algorithms perform better than $k$-means, and $k$-groups by first variation had slightly higher agreement measures than Hierarchical $\xi$ on this data set.

In summary, our proposed $k$-groups method can be recommended for all types of unsupervised clustering problems with pre-specified number of clusters, because performance was typically comparable to or better than $k$-means. $k$-groups has other advantages and it is a more general method. It can be applied to cluster feature vectors in arbitrary dimension and the index $\alpha$ can be chosen to handle very heavy tailed data with non-finite expected distances. We have developed and applied an updating formula analogous to Hartigan and Wong, which has been implemented in R and provided in an R package \emph{kgroups} \cite{kgroups}. Our algorithms could also be implemented in Python, Matlab or other widely used languages.

\bibliographystyle{chicago}
\bibliography{reference}
\end{document}